\documentclass[12pt]{article}
\usepackage{amssymb}
\usepackage{amsmath}
\usepackage{cases}
\usepackage{booktabs}
\usepackage{latexsym, graphicx}
\usepackage{array}
\usepackage{multirow}
\usepackage{longtable}
\usepackage{setspace}
\usepackage{program}
\usepackage{tabu}
\usepackage{algpseudocode}
\usepackage[ruled, longend]{algorithm2e}
\usepackage{float}
\usepackage{lipsum}
\usepackage{mathtools}
\usepackage{subfig}

\usepackage[round]{natbib}
\usepackage[bookmarks=true]{hyperref}
\usepackage{url}

\usepackage[pagewise]{lineno}
\def\qed{\hfill$\Box$}

\usepackage{titlesec}
\titleformat{\section}
  {\normalfont\normalsize\bfseries}{\thesection}{1em}{}
\titleformat{\subsection}
  {\normalfont\normalsize\bfseries}{\thesubsection}{1em}{}
\titlespacing{\section}{0pt}{12pt}{6pt}
\titlespacing{\subsection}{0pt}{12pt}{6pt}
\titlespacing{\subsubsection}{0pt}{10pt}{0pt}

\usepackage{siunitx}

\usepackage{fancyhdr}
\fancypagestyle{plain}{%
\fancyhead{} % get rid of headers on plain pages
}

\newtheorem{definition}{Definition}
\newtheorem{proposition}{Proposition}

\newcommand{\Rea}{{\mathbb R}}

\newcommand{\Ss}{{\mathbb S}}

\newcolumntype{C}[1]{>{\centering\let\newline\\\arraybackslash\hspace{0pt}}m{#1}}

\title{\Large \textbf{Fair Cost Allocation for Ridesharing Services -- Modeling, Mathematical Programming and an Algorithm to Find the Nucleolus}} 
\author{Wei Lu\thanks{Amazon.com, Inc., \url{theweilu@amazon.com}. This work was done prior to the author joining Amazon.}  and Luca Quadrifoglio \thanks{Zachry Department of Civil Engineering, Texas A\&M University, \url{lquadrifoglio@civil.tamu.edu }, the corresponding author. }}

\date{December 31, 2018}  
    
\begin{document}
\begin{spacing}{1.2}
\maketitle
%\pagebreak
\begin{abstract}
%Ridesharing services, whose aim is to gather travelers with similar itineraries and compatible schedules, are able to provide substantial environmental and social benefits through reducing the use of private vehicles. When the operations of a ridesharing system is optimized, it can also save travelers a significant amount of transportation cost. The economic benefits associated with ridesharing in turn attract more travelers to participate in ridesharing services and thereby improve the utilization of transportation infrastructure capacity. 

%Although the prospect of ridesharing services is very promising from both the societal and individual point of view, ridesharing success stories are still in short supply.  
%Designing an optimal operations strategy for ridesharing systems is a very challenging task. 

This paper addresses one of the most challenging issues in designing an efficient and sustainable ridesharing service: ridesharing market design. We formulate it as a fair cost allocation problem through the lens of the cooperative game theory. A special property of the cooperative ridesharing game is that its characteristic function values are calculated by solving an optimization problem. Several concepts of fairness are investigated and special attention is paid to a solution concept named nucleolus, which aims to minimize the maximum dissatisfaction in the system. Due to its computational intractability, we break the problem into a master-subproblem structure and two subproblems are developed to generate constraints for the master problem. We propose a coalition generation procedure to find the nucleolus and approximate nucleolus of the game. Experimental results showed that when the game has a non-empty core, in the approximate nucleolus scheme the coalitions are computed only when it is necessary and the approximate procedure produces the actual nucleolus. And when the game has an empty core, the approximate nucleolus is close to the actual one. Regardless of the emptiness of the game, our algorithm needs to generate only a small fraction ($1.6\%$) of the total coalition constraints to compute the approximate nucleolus. The proposed model and results nicely fit systems operated by autonomous vehicles. %, and the approximate nucleolus is close to the actual nucleolus. 
%From a societal perspective, an optimized ridesharing service can potentially provide substantial system-wide travel cost saving and vehicle-trip saving. However, the system-level optimal solution might not completely align with individual participant interest. 
\\
%\\
%\\
%This paper studies link travel time estimation using entry/exit time stamps of trips on a steady-state transportation network. We propose two inference methods based on the likelihood principle, assuming each link associates with a random travel time. The first method considers independent and Gaussian distributed link travel times, using the additive property that trip time has a closed-form distribution as the summation of link travel times. We particularly analyze the mean estimates when the variances of trip time estimates are known with a high degree of precision and examine the uniqueness of solutions. Two cases are discussed in detail: one with known paths of all trips and the other with unknown paths of some trips. We apply the Gaussian mixture model and the Expectation-Maximization (EM) algorithm to deal with the latter. The second method splits trip time proportionally among links traversed to deal with more general link travel time distributions such as log-normal. This approach builds upon an expected log-likelihood function which naturally leads to an iterative procedure analogous to the EM algorithm for solutions. Simulation tests on a simple nine-link network and on the Sioux Falls network respectively indicate that the two methods both perform well. The second method (i.e., trip splitting approximation) generally runs faster but with larger errors of estimated standard deviations of link travel times.

\end{abstract}
\pagebreak

\section{Introduction}
Ridesharing services, whose aim is to gather travelers with similar itineraries and compatible schedules, are able to provide substantial environmental and social benefits through reducing the use of private vehicles. When the operations of a ridesharing system are optimized, it can also save travelers a significant amount of transportation cost. The economic benefits associated with ridesharing in turn attract more travelers to participate in ridesharing services and thereby improve the utilization of transportation infrastructure capacity. %In addition, our model and solutions could also be applied to a driverless ridesharing system, composed of autonomous vehicles, almost a reality in a near foreseeable future.

An optimized ridesharing service is usually designed to minimize the system-wide travel cost. This is beneficial in the society's point of view, assuming each agent accepts the system's assignment. This is, however, a strong assumption considering agents might form their own ridesharing groups if they find doing so is more of their own interest. 
%In Section~\ref{chap: RSP} the ridesharing optimization problem is modeled as a mixed-integer program and the optimal solution is found to minimize the system-wide travel cost. This is beneficial in the society's point of view, assuming each agent accepts the system's assignment. This is, however, a strong assumption considering agents might form their own ridesharing groups if they find doing so is more of their own interest. 

Recall that the agents of ridesharing system participate in this system in the hope of saving travel cost. So it is up to the ridesharing service provider to decide how the travel cost would be shared among customers after a ridesharing plan is proposed and accepted by the customers. This is a non-trivial task because if the agents find the cost allocation scheme unfair, they may leave the system and form their own ridesharing group in the long run. This fair cost-allocation situation is critical to the sustainability of a ridesharing system and thus is the motivation of the study in this paper.

The ridesharing cost allocation problem is modeled as a cooperative game. Cooperative game theory, due to its close relation to combinatorial optimization, has drawn significant attention of the operations research community. Since its introduction by \citep{vonNeumann1944}, cooperative game theory has developed several solution concepts that aim to resolve the benefits (cost) allocation issues among cooperative players. In this paper, we are primarily concerned with a particular cost allocation solution concept - the nucleolus. The nucleolus of a cooperative game has several nice properties. Intuitively, it is a solution to the cost allocation problem that minimizes the maximal dissatisfaction among the customers. 

The concept of nucleolus was first suggested by \citep{nucleolus1969} and since then was developed by \citep{ShapleyBalance} and \citep{Shapley1979Geometry}. Although the nucleolus has several game theoretic virtues, the computation of nucleolus is very difficult. In fact, for a n-player game, as the size of the characteristic function grows exponentially with the number of players, any enumeration algorithm that computes the nucleolus that requires the entire information of the characteristic function takes $O(2^n)$ time, assuming the characteristic function is readily available. Moreover, as will be shown in later section, finding the characteristic function value of ridesharing game involves solving an optimization problem related to the Traveling Salesman Problem (TSP), which is NP-hard itself. This means the computation of the nucleolus of ridesharing game can easily become intractable and more efficient algorithm needs to be developed. 

In this paper, we propose a nucleolus-finding algorithm for the ridesharing game by successively solving a number of linear and integer programs. The linear programming (LP) problem for nucleolus calculation was first studied by \cite{kopelowitz1967computation} and stimulated several LP-based algorithms for nucleolus computation. \cite{Dragan} suggested an algorithm for computing the nucleolus by generating the minimal balanced sets of the player set.  Our nucleolus-finding procedure combine the LP-based algorithm with the constraint generation framework proposed in \citep{Hallefjord}, such that the explicit information of the characteristic function of a ridesharing coalition is only computed when it is ``dissatisfied''. In this way the computational burden is significantly reduced. 

Note that the constraint generation approach was first proposed in~\citep{Gomory1961} and was successfully applied to solving the cutting stock problem. Utilizing a similar idea, \citep{Basic_VRG} studied the basic vehicle routing game (VRG) in which a fleet with homogeneous capacity are available. The authors analyzed the properties of this game and proposed a nucleolus-finding procedure based on coalition generation. \citep{Engevall2004} generalized the model of \citep{Basic_VRG} to consider vehicles with heterogeneous capacities and studied a real-world case based on their model. 

%Numerous problems that arise in the economics and society can be modeled via the lens of cooperative game theory. Among these problems we are particularly interested in the \emph{cost-allocation problem} whose objective is to find the best cost-allocation method such that the total cost of the game is ``fairly" divided among the players in the game. It is not hard to see that the ridesharing game (RSG) naturally fits in this context. Cost-allocation problems that are related to the RSG include the traveling-salesman game (TSG) \citep[see]{Dror1990, EngevallTSG}, the assignment game \citep{ShapleyShubik}, the bin-packing game and the knapsack game \citep{Dror1990}, and the vehicle routing game \citep[see]{Basic_VRG, Engevall2004}. 

The recent prosperity of ridesharing services has spurred a growing attention from the research community. There are some game-theoretic studies, either from a cooperative or non-cooperative perspective, that focus on the mechanism and stability of ridesharing recommendations. \cite{Shen:2016} proposed an online ridesharing mechanism that satisfy ex-post incentive compatibility, individual rationality, and budget-balance in a non-cooperative context. \cite{7365485} designed a double auction based discounted trade reduction mechanism for dynamic ridesharing pricing that is individual rational, incentive compatible, budget balancing and has a larger trading volume. \cite{DBLP:journals/corr/GopalakrishnanM16} studied the costs and benefits of dynamic ridesharing by introducing the notion of sequential individual rationality and sequential fairness. \cite{wang_stable_2018} introduced the concept of stable matches, understanding and addressing the gap and trade-off between the wholistic optimal matchings of and the optimal matchings from individual rider's perspective. From a cooperative game theoretic perspective, \cite{BISTAFFA201786, Bistaffa:2015, Bistaffa:2015:SRF:2887007.2887092} tackled the coalition formation and payment allocation for the so-called social ridesharing problem, in which the feasible coalitions of a set of commuters are restricted by a social network represented by a graph. The authors focused on the solution concept of kernel-stable payments. 

To sum up, this paper advances the state of the art as the following. In contrast with the previous related work in which the existence of several restrictions (fixed driver/rider roles, possible coalitions limited by social network, etc.) significantly reduces the search space, our model has no such restrictions and therefore is more general. Second, to our best knowledge, our work is the first attempt that computes the nucleolus which is provably the most stable payment allocation scheme, compared to other concepts such as kernel and core, in the context of ridesharing cost allocation.

This paper is organized as follows. In Section~\ref{sec: formulation}, a formulation of the ridesharing cost allocation problem is developed from a game theory perspective and the properties of the characteristic function are analyzed. Section~\ref{sec: fairness} discusses the fairness issues in the ridesharing game regarding the core and the nucleolus. A coalition generation scheme is then developed to compute the nucleolus. The constraint generation subproblem is explicitly formulated by a mathematical formulation related to the ridesharing optimization problem. In Section~\ref{sec: exp}, numerical experiments are conducted and the performance of the proposed nucleolus procedure is evaluated. Finally, conclusions and future research ideas are presented in Section~\ref{sec: conclusions}.

\section{Ridesharing Optimization Problem From A Game Theory Perspective}
\label{sec: formulation}
Consider a set of ridesharing participants and denote it by $N$. Each participant wants to travel from her origin $s_i$ to her destination $t_i$. Each participant can potentially be the driver. Denote the capacity of a vehicle by $Q$. Consider the subsets of participants that do not exceed the vehicle capacity, i.e. $\mid S \mid\le Q$. For each such participant subset $s$, assume the \textit{feasible} route with minimum cost is known. Here by feasible it means the following conditions are met

\begin{enumerate}
	\item the route $r$ starts from an agent $d's$ origin and ends at his destination.
	\item Let $s_{-d} = s\backslash\{d\}$. For every agent $i\in s_{-d}$, $s_i$ precedes $t_i$ in $r$.
\end{enumerate}

Denote by $c_r$ the cost of such a feasible route and by $R$ the set of feasible routes with minimal cost. Let $a_{ir}=1$ if participant $i$ (both $s_i$ and $t_i$) is served by route $r$ and $0$ otherwise. The ridesharing optimization problem (RSP) can be formulated as

\begin{align}
	\label{obj}
	\text{(RSP)} \quad z = \quad \text{min} \quad &\sum_{r\in R} c_{r}x_{r}  \\
	\label{RSP: once}
	\text{s.t.} \quad &\sum_{r\in R} a_{ir}x_{r}= 1,  \quad i\in N \\
	\label{arriveDepot}
	&x_r \in \{0, 1\}, r\in R 
\end{align}

In the formulation $x_r=1$ if feasible route $r$ is selected and $0$ otherwise. Constraints \eqref{RSP: once} guarantee that each participant is covered by exactly one route. Note that the coefficient $c_r$ in the objective function is obtained by finding the minimal cost route that covers the participants for which $a_{ir}=1$, that is, by finding the solution to the corresponding TSP with precedence constraints. 

It is noted that this formulation is characterized by its large number of columns. Therefore, this formulation is practically solvable by a column generation solution method. Similar approaches were successfully applied to the vehicle routing problems (VRP)~\citep[see]{Balinski1964, VRPTW92}. When we solve the RSP with a column generation approach, it is of our interest to reduce the number of columns. We show this is possible as follows. 

We first introduce the definition of the \emph{profitable} ridesharing route. 

\begin{definition}[profitability]
	Denote by $r(S)$ the corresponding minimum cost feasible route of participant subset $S$. $r(S)$ is \emph{non-profitable} if there exists two non-empty subsets $S_1 \cup S_2 = S, S_1\cap S_2 = \emptyset$ such that $c_{r(S)} > c_{r(S_1)} + c_{r(S_2)}$. A route is defined profitable otherwise. 
\end{definition}

Intuitively, a shared-ride route becomes non-profitable if by ridesharing the participants end up spending more money on the transportation cost. %This definition is better illustrated by the following example. 

%\begin{figure}[h]

%\end{figure}

The following proposition shows that we only need to consider a subset of the columns when solving RSP. 

\begin{proposition}
\label{prop: profitable}
	Let $X=\{x_{r_1}, x_{r_2}, \ldots, x_{r_m}\}$ be an optimal solution to RSP, i.e. $x_{r_i} = 1, i=1,\ldots, m$. Then $r_i, \forall i=1,\ldots, m$ must be a profitable route. 
\end{proposition}

\begin{proof}
	Proof by contradiction. Let $X=\{x_{r_1}, x_{r_2}, \ldots, x_{r_m}\}$ be an optimal solution to RSP. Suppose there exists $i^*$ such that $r_{i^*}$ is a non-profitable route. Let $S^*$ be the corresponding participants that are covered by this route. Then by definition there must be two non-empty subsets $S_1^* \cup S_2^* = S^*, S_1^*\cap S_2^* = \emptyset$ such that $c_{r(S^*)} > c_{r(S_1^*)} + c_{r(S_2^*)}$. Since all the customers that are covered by $r_{i^*}$ are also covered by $r(S_1^*)$ and $r(S_2^*)$, we can get a new feasible solution $X'$ to RSP by substituting $x_{r_{i^*}}=1$ with $x_{r(S_1^*)}=1$, $x_{r(S_2^*)}=1$ and $x_{r_{i^*}}=0$ while keep all the other $x$ variables unchanged. This feasible solution has a strictly less cost than $X$. Contradiction. \qed
\end{proof}

From a game theory perspective, we denote each ridesharing participant, $i\in N$, by a \textit{player} and each subset of participants, $S\subseteq N$, by a \textit{coalition}. 

The ridesharing cost allocation problem is the problem of finding a ``fair'' cost allocation scheme for the ridesharing optimization problem (RSP). 

A cooperative ridesharing game is defined by specifying a travel cost for each coalition. The game is defined by a ridesharing group $S$, and a \textit{characteristic function} $c(S):2^{S} \to \Rea$ from the set of all possible coalitions (sub-ridesharing group) of players in $S$ to a set of payment schemes satisfying $c(\O)=0$. Here $2^S$ denotes the power set of $S$. In the context of RSP game the characteristic function can be seen as the travel cost occurring if coalition $S\subseteq N$ is formed. Each coalition can be defined by a binary vector $s$ as 

\begin{equation*}
s_i = \begin{dcases}
		1, \quad \text{if customer $i$ is a member of the coalition}, \\
		0, \quad \text{otherwise},
	    \end{dcases}
\end{equation*}

Define $c$ as the objective value of a mathematical program. For all coalitions $S\subseteq N, S\neq \emptyset$, let $c(S)$ be the solution to the following mathematical program

\begin{align}
	\label{obj}
	c(S) = \quad \text{min} \quad &\sum_{r\in R} c_{r}x_{r}  \\
	\label{once}
	\text{s.t.} \quad &\sum_{r\in R} a_{ir}x_{r}= s_i,  \quad i\in N \\
	\label{arriveDepot}
	&x_r \in \{0, 1\}, \quad r\in R 
\end{align}

Intuitively, $c(S)$ represents the cost of an optimal route that covers the players in $S$, i.e. the players for which $s_i=1$. 

%It is noted that this program is very similar to VRP. In fact, the columns of this program can be reduced in a similar fashion as VRP. This is stated in the following proposition. 

%\begin{proposition}
%	Let $X=\{x_{r_1}, x_{r_2}, \ldots, x_{r_m}\}$ be an optimal solution to $C(S)$, i.e. $x_{r_i} = 1, i=1,\ldots, m$. Then $r_i, \forall i=1,\ldots, m$ must be a profitable route. 
%\end{proposition}

%\begin{proof}
%	Similar to the proof of Proposition~\ref{prop: profitable}. \qed
%\end{proof}

When studying a cooperative game, it is of great interest to study the properties of its characteristic function. Assuming that the singleton coalitions have a positive cost, we prove that the characteristic function of the RSP game is \emph{monotonic} and \emph{subadditive}. Interested readers can find the proof in the Appendix. It is noteworthy that subadditivity implies larger coalitions save more. So it is always beneficial to include more people to participant in ridesharing and this is a desirable property of the RSP game. 

Denote a coalition $S$ whose cardinality is smaller than the vehicle capacity ($\mid S\mid \le Q$) as a \textit{feasible coalition} and otherwise as an \textit{infeasible coalition}. Denote by $\Ss$ the set of feasible coalitions. Then we get
\begin{equation*}
	c(S) = c_r, \quad \forall r\in R \text{ and } S\in \Ss \text{ such that } a_{ir} = s_i, i \in N.
\end{equation*}

In addition, denote a coalition such that $\sum_{i\in S}c(i)\ge c(S)$ by a \textit{profitable coalition} and otherwise by a \textit{nonprofitable coalition}.  

\begin{figure}[ht]
\centering
\subfloat[Profitable Coalition, $c(S)=13$]{%
	\includegraphics[width=2in]{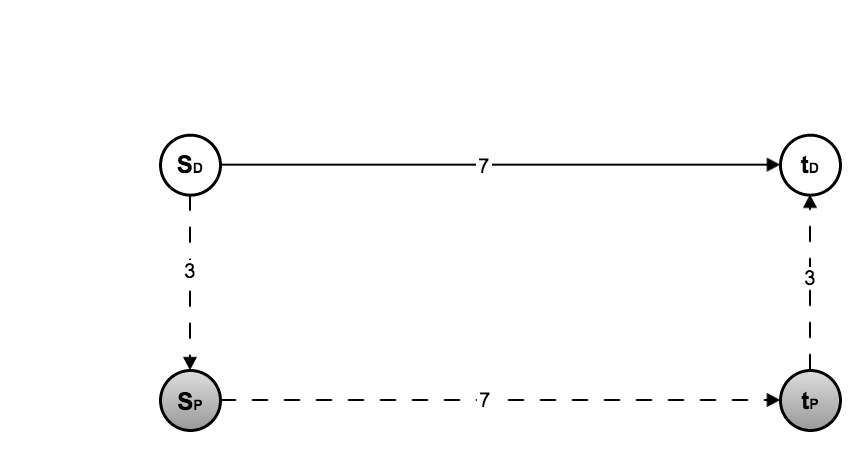}
	\label{p}}
\quad \quad
\subfloat[Nonprofitable Coalition, $c(S)=11$]{%
	\includegraphics[width=2in]{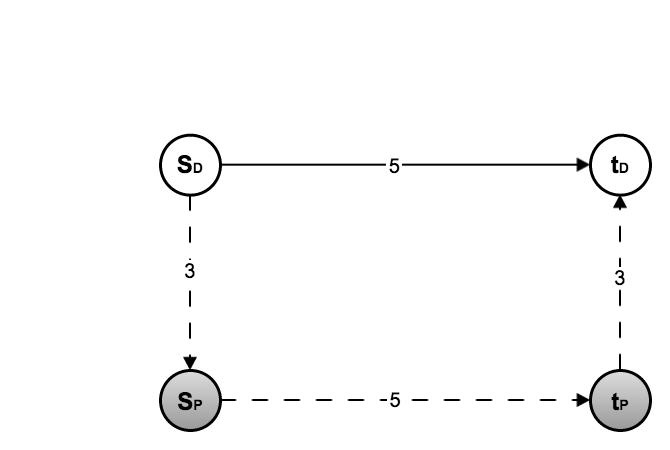}
	\label{np}}
%\subfigure[Sleepy Smiley]{%
%\includegraphics...
%\label{fig:subfigure3}}
%\quad
%\subfigure[Angry Smiley]{%
%\includegraphics...
%\label{fig:subfigure4}}
%
\caption{Profitable vs. nonprofitable coalition}
\label{pvsnp}
\end{figure}

Figure~\ref{pvsnp} gives an example where forming a coalition will not always produce desirable results: instead of reducing total transportation cost as Figure~\ref{pvsnp}\subref{p}, Figure~\ref{pvsnp}\subref{np} actually increases the total cost, meaning it doesn't make much sense to form such a coalition. In this case the players are better off on their own. Note that the profitability of forming a coalition in a large extent depends on the relative geo-locations of the players.

\section{Fairness and Stability in RSP Game}
\label{sec: fairness}

\subsection{The core and the nucleolus}
\label{subsec: prop}
\subsubsection{The core}

Let $y_i$ be the cost allocated to agent $i$, $i\in N$. The \textit{core} of the RSP game is the set of the cost allocation plans $y$, such that 

\begin{align}
	\label{impu}
	\sum_{i\in N}y_i &= c(N), \\
	\label{CDI}
	\sum_{i\in S}y_i &\le c(S), \forall S\subset N.
\end{align}

The above inequalities can be interpreted as no single player or coalition should make a payment that is greater than their cost on their own. A cost allocation scheme that is in the core is a good allocation as no coalition has the incentive to leave the grand coalition. An inequality in~\eqref{CDI} is called a \textit{core defining inequality} (CDI). 

%The core of a RSP game can be empty or non-empty. 

%\begin{proposition}\citep{Basic_VRG}
%	The core of the basic VRP game is nonempty if and only if $\bar{z} = z$.
%\end{proposition} 

%\begin{proof}

%\end{proof}

It is observed that the number of CDIs is in the scale of $O(2^N)$. As will be shown in later sections, in order to find the core and the nucleolus efficiently, it is important and of our great interest to reduce the number of CDIs. This is possible through the following propositions.   

\begin{proposition}
\label{prop: profitableCoalition}
	Any CDI with a nonprofitable coalition $S, S\neq N$, is not needed in \eqref{CDI}.
\end{proposition}

\begin{proof}
	Consider any nonprofitable coalition $\hat{S}$, $\hat{S}\neq N$. Denote by $\{1, 2, \ldots, m\}$ the players in $\hat{S}$. By definition of nonprofitable coalition we have $\sum_{i\in \hat{S}}c(i)< C(\hat{S})$. Note that all individual players are also singleton coalitions. It follows that 
	\begin{equation*}
		y_i\le c(i), \forall i\in \hat{S} \quad \Longrightarrow \quad \sum_{i\in \hat{S}} y_i\le \sum_{i\in \hat{S}} c(i)<c(\hat{S}). 
	\end{equation*}  \qed
\end{proof}

\begin{proposition}%\citep{Basic_VRG}
\label{prop: feasible}
	For a RSP game with non-empty core, any CDI with an infeasible coalition $S, S\neq N$, is not needed in \eqref{CDI}.
\end{proposition}

\begin{proof}
	Similar to \citep{Basic_VRG}, let $\hat{S}$, $\hat{S}\neq N$ be an infeasible coalition. Denote by $\{r_1,\ldots, r_m\}$ the corresponding optimal routes and $\{s_1, \ldots, s_m\}$ the disjoint feasible coalitions corresponding to the optimal routes. Since we have $\sum_{j=1}^m \sum_{i\in S_j} y_i = \sum_{i\in \hat{S}} y_i$ and $\sum_{j=1}^m c(S_j) = c(\hat{S})$, then we have the following 
	\begin{equation*}
		\sum_{i\in S_j} y_i\le c(S_j), \forall j=1,\ldots, m \quad \Longrightarrow \quad \sum_{i\in \hat{S}} y_i\le c(\hat{S}). 
	\end{equation*} 	\qed
\end{proof}

From Proposition~\ref{prop: profitableCoalition} and Proposition~\ref{prop: feasible} we have 
	\begin{equation*}
		C = \{ y \mid \sum_{i\in S} y_i \le c(S), S\in \mathbb{S}; \quad \sum_{i\in N} y_i = c(N) \}.
	\end{equation*}

Thus, when the core of the RSP game is non-empty, the only characteristic function values of our interests are those corresponding to profitable and feasible coalitions. This, as will be stated in later sections, reduces the size of the coalition-generating subproblem dramatically. Note that the calculation of $c(S)$ for a coalition $S$ is equivalent to solving the corresponding traveling salesman problem with pick-up and drop-off constraints (TSPPD) for the customers for which $s_i=1$. 

\subsubsection{The nucleolus}
%Note that the previous definitions in Section~\ref{sec: background} are based on value game. 
In a cooperative ridesharing game, in which players share travel cost, allocations are the payments each player need to pay. That is to say, cooperative ridesharing game (RSP game) is a cost game.

Let $c:2^{S} \to \Rea$ denote the cost characteristic function of a cooperative ridesharing game. Then the function gives the amount of collective cost a group of players need to pay through forming a coalition. In an RSP game, the \emph{excess} of $y$ for a coalition $S\subseteq N$ is defined as $e(y, S) = c(S)-\sum_{i\in S}y_i$ and measures the amount of cost-saving of coalition $S$ in the allocation $y$, compared to $c(S)$. Note that when $e(y, S)$ is negative, it means the sum of the cost of $S$ in the allocation $y$ must exceed $c(S)$. Thus $e(y, S)$ measures the dissatisfaction of $S$ in the allocation $y$. Recall that the core is defined as the set of imputations such that $c(S)\ge \sum_{i\in S}y_i$ for all coalitions $S$, then we have that an imputation $y$ is in the core if and only if all its excesses are positive or zero. Denote by $\theta(y)\in \Rea^{2^N}$ the excess vector of $y$ whose elements $c(S)-\sum_{i\in S}y_i$ are arranged in non-decreasing order, that is, $\theta_i(y)\le \theta_j(y), \forall i<j$. Then a cost allocation vector $y$ is in the core if and only if it is efficient and $\theta_1(y)\ge 0$. Consider the lexicographic ordering of excess vectors: for two payment vectors $x, y$, we say $\theta(x)$ is lexicographically greater than $\theta(y)$ if $\exists k$ such that $\theta_i(x)=\theta_i(y), \forall i<k$ and $\theta_k(x)>\theta_k(y)$. Denote this ordering by $\theta(x)\succ \theta(y)$. 

\begin{definition}[Nucleolus] The nucleolus of a RSP game is the lexicographically maximal imputation. Denote the nucleolus by $y$ and let $\bar{Y}$ be the set of imputations, then we have
	\begin{equation*}
		\theta(y)\succeq \theta(y'), \quad \forall y'\in \overline{Y}\backslash y. 
	\end{equation*}
\end{definition}

Intuitively, the nucleolus is by definition the fairest cost allocation plan because it minimizes the maximal dissatisfaction of all the coalitions in the ridesharing system. Remember that a coalition is a subset of all the ridesharing groups. As a result, on the condition that the core is nonempty, the nucleolus is the \emph{center} of the core, and a ridesharing system that implements the nucleolus as the cost sharing plan is provably the most stable. It is of our interest to investigate the non-emptiness of the core of RSP game because the definitions of nucleolus and the core are related. In fact, the following example shows that the core of RSP game may be empty. 

\begin{figure}[h] 
\centering
\includegraphics[width=4in]{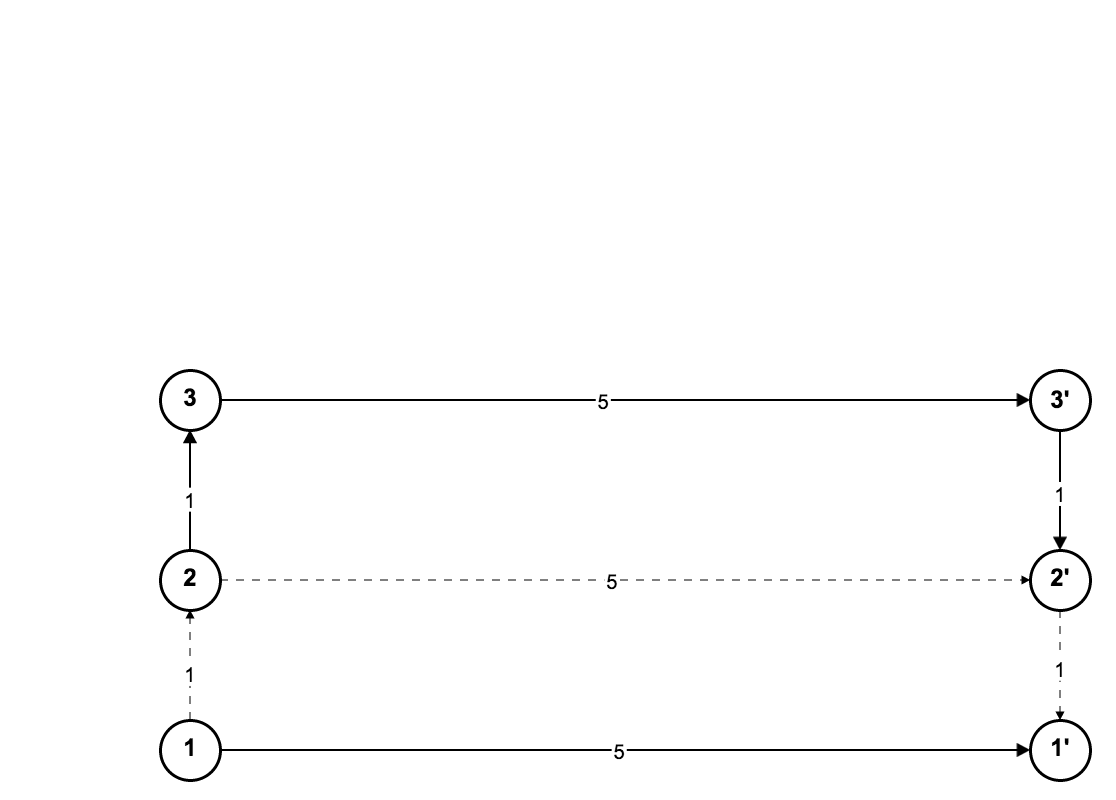}
\caption{A three player example}
\label{fig: empty_vs_nonempty}
\end{figure}

% \begin{figure}[ht]
% \centering
% \subfloat[Empty core]{%
% 	\includegraphics[width=2.5in]{figures/empty_core.png}
% 	\label{empty}}
% \quad
% \subfloat[Nonempty core]{%
% 	\includegraphics[width=2.5in]{figures/nonempty_core2.png}
% 	\label{nonempty}}
% \caption{A three player example}
% \label{fig: empty_core}
% \end{figure}

The transportation costs of three players $1,2,3$ are given in Figure~\ref{fig: empty_vs_nonempty}. Assuming that each player's vehicle has a capacity of one extra passenger seats, the characteristic function of this 3-player game is then defined by $c(\{1\}) = c(\{2\}) = c(\{3\}) = 5$, $c\{1,2\} = c(\{2,3\}) = 7$ (e.g., $1-2-2'-1'$), $ c(\{1,3\}) = 9$ (e.g., $1-3-3'-1'$) and $c(\{1,2,3\}) = 12$. The optimal route configuration is, for example, $2-3-3'-2'$ and $1-1'$. We show the calculation of nucleolus of this example in Table~\ref{cal: nucleolus}.

As an initial guess, we try $(4, 4, 4)$. In Table~\ref{cal: nucleolus}, we find that the minimum excess happens at coalition $(1, 2)$ and $(2, 3)$. These are the coalitions with maximum dissatisfaction. To improve this, we must increase both $y_1$ and $y_3$. This involves decreasing $y_2$, and will decrease the excess for $(1, 3)$ at the same rate. Also note that player $1$ and player $3$ have symmetrical roles in this game, thus we can conclude that the best scenario occurs when the excesses for $(1, 2)$, $(2, 3)$ and $(1, 3)$ are all equal. Solving the equations,
\begin{align*}
	&y_3-5 = y_1-5 = y_2-3 \\
	&y_1+y_2+y_3=12,
\end{align*}

we find the nucleolus of this game is $y = (4 \frac{2}{3}, 2\frac{2}{3}, 4 \frac{2}{3})$ and $C = \emptyset$ (note that $c(1,2)<y_1+y_2$).

%Interestingly, if we change the geographical location of player 1 to Figure~\ref{fig: empty_core}\subref{nonempty}, then we obtain a game with nonempty core. In this case, the characteristic function is defined in the following fashion. For the singleton coalitions and the grand coalition, characteristic function values remain the same. However, for the 2-coalitions, $c(\{1,2\}) = 5$, $c(\{1,3\}) = c(\{2,3\}) = 6$. The optimal route configuration is the same as Figure~\ref{fig: empty_core}\subref{empty}, yet the core of this game is nonempty. 

\begin{table}[]
\centering
\caption{Calculation of nucleolus - empty core}
\label{cal: nucleolus}
\begin{tabular}{ccccc}
S       & c(S) & e(y, S)             & (4, 4, 4) & $(4\frac{2}{3}, 2\frac{2}{3}, 4\frac{2}{3})$ \\
1       & 5    & $5-y_1$             & 1         & $\frac{1}{3}$                                \\
2       & 5    & $5-y_2$             & 1         & $2\frac{1}{3}$                               \\
3       & 5    & $5-y_3$             & 1         & $\frac{1}{3}$                                \\
1 and 2 & 7    & $7-y_1-y_2 = y_3-5$ & -1        & $-\frac{1}{3}$                               \\
2 and 3 & 7    & $7-y_2-y_3 = y_1-5$ & -1        & $-\frac{1}{3}$                               \\
1 and 3 & 9    & $9-y_1-y_3 = y_2-3$ & 1         & $-\frac{1}{3}$                              
\end{tabular}
\end{table}

Interestingly, if we increase the capacity of the vehicles to two extra seat, then we obtain a game with a nonempty core. In this case, the characteristic function is defined in the following fashion. For the singleton coalitions and the 2-coalitions, characteristic function values remain the same. However, for the grand coalition $c(\{1,2,3\}) = 9$. The optimal route configuration is, for example, $1-2-3-3'-2'-1'$. The calculation of nucleolus of this example is shown in Table~\ref{cal: nucleolus-2}.

\begin{table}[]
\centering
\caption{Calculation of nucleolus - nonempty core}
\label{cal: nucleolus-2}
\begin{tabular}{ccccc}
S       & c(S) & e(y, S)             & (3, 3, 3) & $(\frac{11}{3}, \frac{5}{3}, \frac{11}{3})$ \\
1       & 5    & $5-y_1$             & 2         & 4/3                                         \\
2       & 5    & $5-y_2$             & 2         & 10/3                                        \\
3       & 5    & $5-y_3$             & 2         & 4/3                                         \\
1 and 2 & 7    & $7-y_1-y_2 = y_3-2$ & 1         & 5/3                                         \\
2 and 3 & 7    & $7-y_2-y_3 = y_1-2$ & 1         & 5/3                                         \\
1 and 3 & 9    & $9-y_1-y_3 = y_2$   & 3         & 5/3                          
\end{tabular}
\end{table}

Initially, we try $(3, 3, 3)$. In Table~\ref{cal: nucleolus-2}, we find that the minimum excess happens at coalition $(1, 2)$ and $(2, 3)$. These are the coalitions with maximum dissatisfaction. To improve this, we must increase both $y_1$ and $y_3$. This involves decreasing $y_2$, and will decrease the excess for $(1, 3)$ at the same rate. Also note that player $1$ and player $3$ have symmetrical roles in this game, therefore we can conclude that the best scenario we can achieve happens when the excesses for $(1, 2)$, $(2, 3)$ and $(1, 3)$ are all equal. Solving the equations,
\begin{align*}
	&y_3-2 = y_1-2 = y_2 \\
	&y_1+y_2+y_3=9,
\end{align*}

we find the nucleolus of this game is $y = (11/3, 5/3, 11/3)$. Here, $C\neq \emptyset$ and the nucleolus of this game is the center of the core.

%These observations can be generalized below.

%\begin{proposition}\citep{Basic_VRG}
%	Let $R$ be an optimal route configuration of $N$, i.e. $R$ consists of routes serving the players of feasible disjoint coalitions $S_1, S_2, \ldots, S_m$. Then we have
%	\begin{equation*}
%		\sum_{j\in S_r} y_j = c(S_r), \quad \forall y\in C \text{ and } 1\le r\le m . 
%	\end{equation*}
%\end{proposition}

%\begin{proof}
%	Note that $c(N) = \sum_{r=1}^m c(S_r)$ and $\sum_{r=1}^m\sum_{j\in S_r} y_j = \sum_{j\in N} y_j, \forall y\in R^n$. In addition any $y\in C$ satisfies $\sum_{j\in N} y_j = c(N)$ and $\sum_{j\in S_r} y_j, \forall 1\le r\le m$. Therefore, $c(N) = \sum_{j\in N} y_j = \sum_{r=1}^m\sum_{j\in S_r}y_j\le \sum_{r=1}^m c(S_r) = c(N)$. It implies that all inequalities must be equalities for this equation to hold, meaning $\sum_{j\in S_r} y_j = c(S_r), \quad \forall y\in C \text{ and } 1\le r\le m$.  \qed
%\end{proof}

%Therefore, if the core of RSP game is nonempty, then the cost of an optimal route should be split only among the players served by that route.  

% \begin{proposition}[non-empty core condition]
	
% \end{proposition}

\subsection{An algorithm to find the nucleolus}

%\cite{kopelowitz1967computation}   \cite{Engevall2004}
%points :
%- AV
%- most stable
%- 
As discussed before, finding a nucleolus will ensure the ridesharing system implements the provably fairest cost allocation plan to the users, which also ensures the stability of the system. A procedure to calculate the nucleolus is developed in this paper and with sufficient amount of computational resources, a system that is able to handle realistically large-scale ridesharing service system can be developed. In a nutshell, the proposed algorithm starts with the least core, and continues with lexicographic optimization through iterating between the master problem and the subproblems. We detail this procedure below. 

%How would this search procedure be implemented with current technology and how would it benefit users/operators. Is it something with a possible effect in the real world or it is only a theoretical game affected by unrealistic and un-implementable assumptions?

%Mainly: a clear paragraph to bring readers back to Earth and show them the impact on real systems.

\subsubsection{The master problem}

Since the nucleolus is in the least core, which minimizes the maximal dissatisfaction, we start with the solution to the following maximin problem
\begin{equation*}
	\max_{y\in Y} \min_{\forall S\subset N} (c(S) - \sum_{i\in S}y_i),
\end{equation*}

which can be transformed to a linear program
\begin{align}
	\label{p1}
	\text{($P^1$)} \quad \max \quad &w  \\
	%\label{RSP: once}
	\text{s.t.} \quad &w\le c(S) - \sum_{i\in S} y_i, \quad S\in \mathbb{S} \\
	%\label{arriveDepot}
	&\sum_{i\in N} y_i = c(N).
\end{align}

Notice that the LP program has $O(\mid\mathbb{S}\mid)$ constraints, and computing $c(S), S\in \mathbb{S}$ involves solving the corresponding TSPPD. So the LP program can easily become intractable. We therefore approach this problem with a constraint generation procedure. \cite{Hallefjord} has suggested such an approach for linear programming games. \cite{Basic_VRG} has used a similar approach to solve the vehicle routing problem (VRP) game.

Since before searching for the nucleolus, we should already know the solution to the corresponding RSP, thus the optimal route configuration, we can start $(P_1)$ with the coalitions corresponding to the optimal routes. Besides, the singleton coalitions' cost values are readily available. Denote by $\Omega \in \Ss$ the available coalitions, then $(P_1)$ can be replaced by the following relaxed problem
\begin{align}
	\label{pM1}
	\text{($P_M^1$)} \quad \max \quad &w  \\
	%\label{RSP: once}
	\text{s.t.} \quad &w\le c(S) - \sum_{i\in S} y_i, \quad S\in \Omega \\
	%\label{arriveDepot}
	&\sum_{i\in N} y_i = c(N).
\end{align}

If the solution to ($P_M^1$) is unique, let it be $(y^*, w^*)$, i.e. $\theta_1(y^*)>\theta_1(y'), \forall y'\in Y\backslash \{y^*\}$, then $y^*$ is the nucleolus of the game. If the solution to ($P_M^1$) is not unique, we continue to find the greatest $\theta_2(y)$ among the $y\in Y$ with $\theta_1(y) = w^*$. We continue this process until the solution to the linear program is unique. At stage $t$ the master LP problem to be solved is
\begin{align}
	\label{pt}
	\text{($P_M^t$)} \quad \max \quad &w_t  \\
	%\label{RSP: once}
	\text{s.t.} \quad &w_t\le c(S) - \sum_{i\in S} y_i, \quad S\in \mathbb{S}\backslash \bigcup_{\tau=1}^{t-1} \Gamma_{\tau}, \\
	&w_{\tau} = c(S) - \sum_{i\in S}y_i, \quad S\in \Gamma_{\tau}, \tau = 1,\ldots,t-1,  \\
	%\label{arriveDepot}
	&\sum_{i\in N} y_i = c(N).
\end{align}

The solution to the last program in this series is the nucleolus of this game. Let $\Pi_{t,S}$ be the dual variable corresponding to constraint $w\le c(S) - \sum_{i\in S} y_i$. Let $\Gamma_t$ denote the set of coalitions whose corresponding constraints are \emph{binding}, that is, $\Gamma_t = \{S\in \Ss \bigcup_{\tau=1}^{t-1} \Gamma_{\tau} \mid \Pi_{t,S}^*>0\}$. 

The essential idea of constraint generation approach is trying to find the nucleolus with explicit information of only a small portion of the entire coalition set. This goal is realized by finding the most violated constraint that is not yet included in $\Omega$ via a subproblem after the master problem is solved at each stage. Denote the optimal solution to $(P_M^1)$ by $y^* = (y_1^*, \ldots, y_n^*)$. The constraint that is violated the most, aka the most unhappy coalition given the cost allocation scheme $y^*$, is obtained through solving the following subproblem 
\begin{equation*}
	\text{($P_S$)} \min_{S\in \mathbb{S}\backslash \Omega} c(S) - \sum_{i\in S}y_i^* - w^*
\end{equation*}

%???Let $S^*$ be the solution to $(P_S)$. If $c(S^*) - \sum_{i\in S^*}y_i^* - w^* < 0$, then we include the constraint $w\le c(S) - \sum_{i\in S}y_i^*$ in the master problem and re-solve it.  the objective value of the subproblem is nonnegative. it means all the constraints are satisfied. Then we  and formulate a new master problem of the form $(P_M^t)$ for the next stage $t$. 

This nucleolus-finding procedure for a ridesharing game is developed based on the theories and techniques proposed in \cite{Dragan} and \cite{kopelowitz1967computation} and a general constraint generation framework proposed in \cite{Hallefjord}. The pseudocode of this procedure is given in Algorithm~\ref{algo: nucleolus}. First, at stage $t$ the master LP problem $P_M^t$ is solved and both the primal and dual solutions are returned. Second, a subproblem $P_S$ is solved and the least satisfied constraint ($s^*$) that is not yet included is identified. If $c^* \le 0$, then we include $s^*$ in $\Omega_{INEQ}$ and resolve $P_M^t$ with newly included constraint $w_t\le c(s^*) - \sum_{i\in s^*}y_i^*$. This stage iterates between the master problem and the subproblem until no coalition violates the rationality constraints of the master problem (i.e. $c^* \ge 0$). When this is achieved, we identify the active and binding constraints, reformulate the master problem (modify $\Omega_{INEQ}$ and $\Omega_{EQ}$) and proceed to the next stage ($t = t+1$). This process continues until the solution $y^*$ to the master problem is unique. And this last solution is the nucleolus of the RSP game. Note that in the procedure SP.addCut($s^*$), a cut of the type of inequality (\ref{no-repeat}) is added to the subproblem to prevent the duplication of row associated with coalition $s^*$.
%MP.addRow($s^*, c^*$) adds the constraint that belongs to coalition $s^*$ to the current master LP problem. In the procedure SP.addCut($S$), a cut of the type of inequality (\ref{no-repeat}) is added to the subproblem to prevent the duplication of row associated with coalition $S$. The procedure MP.changeRow($S$) switches the inequality constraint associated with coalition $S$ to an equality constraint.  

The following two subsections will discuss two formulations of subproblem $P_S$. In particular, we will discuss how the second formulation reduces the complexity of the overall procedure by utilizing the aforementioned propositions. 

%\clearpage

\begin{algorithm}[ht!]
    \SetKwInOut{Input}{input}\SetKwInOut{Output}{output}
    \Input{Geolocations of customers}
    \Output{Nucleolus of the ridesharing game}
	\BlankLine
	%Initialize \;
	t $:=$ 1 \;
	$\Omega_{INEQ}$ \;
	STOP $:=$ false \;
	%$S_{out}\leftarrow$ Empty \;
	%$s\leftarrow sde(|z_{00}|^2)$ \;  
	\While{!STOP}{
		%Solve\_LP   \;
		Solve a master problem $P_M^t$    \;
		Solve a subproblem $P_S$  \;
		$c^* = \min_S c(S) - \sum_{i\in S}y_i^* - w^*$ \; 
		$s^* = \arg\min c(S) - \sum_{i\in S}y_i^* - w^*$ \;
		\If{$c^* \le 0$} {
			SP.addCut($s^*$) \;
			$\Omega_{INEQ}:=\Omega_{INEQ}\cup \{s^*\}$  \;
			%MP.addRow($s^*, c^*$) \;
		}
		\Else {
			STOP $:=$ true \;
			%\For{$S\in \Omega_{INEQ}$ {\bf and} $\Pi(S)>0$}{
			\For{every active and binding constraint $s$}{
				STOP $:=$ false \;
				%MP.changeRow($s$) \;
				$\Omega_{INEQ}:=\Omega_{INEQ}\backslash \{s\}$  \;
				$\Omega_{EQ}:=\Omega_{EQ}\cup \{s\}$  \;
			}
			$t := t+1$  \;
		}
		%$t\leftarrow t+1$  \;
  %   	\While{$s>3$}{
		% state$\leftarrow$ unfound \;
		% 	\ForAll{$k\in {0,1,2,3}$}{
		% 		\While{state$=$unfound}{
		% 			$z_{00}'\leftarrow$ top left entry of $HT^{-k}U$ \;
		% 			\If{$sde(|z_{00}'|^2)=s-1$} {
		% 				state$=$found\;
		% 				add $T^kH$ to the end of $S_{out}$\;
		% 				$s\leftarrow sde(|z_{00}'|^2)$\;
		% 				$U\leftarrow HT^{-k}U$\;
		% 			}
		% 		}
		% 	}
  %   	}
    }
	% lookup sequence $S_{rem}$ for $U$ in $\mathbb{S}_3$\;
	% add $S_{rem}$ to the end of $S_{out}$\;
    \caption{Procedure of finding the nucleolus of a cooperative ridesharing game}
    \label{algo: nucleolus}
\end{algorithm}

%\subsubsection{The subproblem}
%\newpage
%\clearpage
%\linebreak
%.
%\subsection{Nucleolus Finding Algorithm}
%Pm, Ps, constraint generation procedure. Explicitly formulate the subproblem as a TSP-related problem. 

\subsection{Coalition generation subproblem -- general}

Recall in the nucleolus-finding scheme described in Algorithm~\ref{algo: nucleolus}, it involves finding the most violated constraint in the subproblem. This is equivalent to finding the ``least satisfied" subset of customers with a given allocation proposal. A general formulation of the subproblem is thus
\begin{align}
	\text{($P_S^0$)} &\min_{S\in \Ss\backslash \Omega} c(S) - \sum_{i\in S}y_i^* - w^* \\
	\label{no-repeat-general}
	&\sum_{\{i \mid s_i^j=0\}} s_i + \sum_{\{i \mid s_i^j=1\}} (1 - s_i) \ge 1, \quad j \mid S_j \in \Omega \\
	&s_i\in \{0, 1\}, i\in N
\end{align}

Constraints~(\ref{no-repeat-general}) are preventing the re-generation of constraints.

Note that calculating $c(S)$ is equivalent to solving the RSP model for customers $i\in S$, i.e. those $s_i = 1$. This implies that we can formulate the subproblem $P_S^0$ explicitly. Denote by $G=(V, E)$ the graph of the RSP game with vertex set $V=V_O \cup V_D \cup \{0\}$ and edge set $E=\{(i,j) \mid i,j \in V, i\neq j\}$. Here vertex $0$ is the ``dummy" depot such that any edge incident with it has a cost of $0$. $V_O(V_D)$ is the origin (destination) vertex set of players in $N$. Each player is associated with a profit (prize) equal to $y_i^*$. The subproblem of the constraint generation procedure is to find a subset of customers in $N$ which maximizes the total prize minus the total cost, while conforming to certain constraints.

%Recall that in Algorithm~\ref{algo: nucleolus} searching the most violated constraint in each iteration is a non-trivial task. Notice that $c(S), S\in \Ss$ is the minimum cost of a feasible route that covers the origin and destination of all the players in $S$, that is, those $s_i=1, i\in N$. Thus we can formulate the subproblem explicitly. Denote by $G=(V, E)$ the graph with vertex set $V=V_O \cup V_D \cup \{0\}$ and edge set $E=\{(i,j) \mid i,j \in V, i\neq j\}$. Here vertex $0$ is the ``dummy" depot such that any edge incident with it has a cost of $0$. $V_O(V_D)$ is the origin (destination) vertex set of players in $N$. Each player is associated with a profit (prize) equal to $y_i^*$. The subproblem of the constraint generation procedure is to find a feasible route in $G$ that maximizes the total prize minus the total cost, while conforming to the following constraints

\begin{align}
	\label{obj}
	\text{($P_S^1$)} \quad \pi = \quad &\text{max} \quad \sum_{k\in N} y_{k}^*s_{k} - \sum_{i\in V}\sum_{j\in V} c_{ij}\lambda_{ij} + w^*  \\
	%\label{leaveDepot}
	%\text{s.t.} \quad &\sum_{j=1}^{n} x_{0j}= m, \\
	%\label{arriveDepot-4}
	%&\sum_{i=n+1}^{2n} x_{i0} = m,  \\
	\label{no-repeat-3}
	&\sum_{\{i \mid s_i^j=0\}} s_i + \sum_{\{i \mid s_i^j=1\}} (1 - s_i) \ge 1, \quad j \mid S_j \in \Omega \\
	%\label{arrive}
	%&\sum_{i=0}^{2n} x_{ij} = 1, \quad j=1, 2, \ldots, n  \\
	%\label{leave}
	%&\sum_{j=0}^{2n} x_{ij} = 1, \quad i=1, 2, \ldots, n  \\
	%\label{noSingleNode}
	%& x_{0i}+x_{i0} \le 1, \quad i=1, 2, \ldots, n  \\
	\label{driverPrecedence}
	& \lambda_{0i} - \lambda_{(i+n) 0} = 0, \quad i=1,2, \dots, n \\
	\label{balance-1}
	%& \sum_{i=0}^{k-1}\lambda_{ik} + \sum_{j=k+1}^{2n}\lambda_{kj} = 2s_k, \quad k\in N \\
	& \sum_{i=0}^{2n}\lambda_{ik} = s_k, \quad k\in N \\
	\label{balance-2}
	%& \sum_{i=0}^{k+n-1}\lambda_{i(k+n)} + \sum_{j=k+n+1}^{2n}\lambda_{(k+n)j} = 2s_k, \quad k\in N \\
	& \sum_{i=0}^{2n}\lambda_{ki} = s_k, \quad k\in N \\
	\label{balance-3}
	& \sum_{i=0}^{2n}\lambda_{i, k+n} = s_k, \quad k\in N \\
	\label{balance-4-1}
	& \sum_{i=0}^{2n}\lambda_{k+n, i} = s_k, \quad k\in N \\
	\label{SEC}
	& u_i - u_j + p\lambda_{ij} \le p-1, \quad 1\le i \neq j \le 2n  \\
	\label{precedence}
	& u_i < u_{i+n}, \quad i=1,2, \ldots, n \\
	%& x_{ij} \in \{0, 1\}, \quad \forall (i,j) \in A \\
	\label{oneDriver-3}
	& \sum_k y_{ik} = 1, \quad i=1, \ldots, 2n; \quad k=1, \ldots, n \\
	\label{ODSameVehicle-3}
	& y_{ik} = y_{(i+n)k}, \quad i=1, \ldots, n; \quad k=1,\ldots, n \\
	\label{driver-3}
	& \lambda_{0i} = y_{ii}, \quad i=1, \ldots, n \\
	\label{bridge1-3}
	& \lambda_{ij} \le M_1(1-z_{ij}), \quad 1\le i \neq j \le 2n \\
	\label{bridge2-3}
	& -M_2z_{ij} \le y_{ik} - y_{jk} \le M_2z_{ij}, \quad 1\le i \neq j \le 2n; \quad k=1, \ldots, n \\
	& \lambda_{ij} \in \{0, 1\}, \quad 0\le i\neq j\le 2n \\
	%& y_i\in INT, \quad i=1,2,\ldots, 2n \\
	& y_{ik} \in \{0, 1\}, \quad i=1, \ldots, 2n; \quad k=1,\ldots, n \\
	& z_{ij}\in \{0, 1\},  \quad 1\le i\neq j\le 2n  \\
	& s_k\in \{0, 1\}, \quad k\in N 
	%& w_i^k\in \{0, 1\}, \quad  i=1,\ldots, n, \quad k=1,2,3 
	%& \text{same vehicle constraints}.
\end{align}

%Easily to see, this program ($P_S^1$) is closely related to the ridesharing optimization model (RSP). 
This problem can be termed as prize-collecting RSP (see~\cite{Balas1989} for an analogy of TSP and prize-collecting TSP). Constraints~(\ref{driverPrecedence}) make sure that a tour starts at origin and ends at destination. Constraints~(\ref{balance-1}) -~(\ref{balance-4-1}) are the continuity constraints. Constraints (\ref{SEC}) are a group of subtour-elimination constraints (SECs) first proposed by~\cite{Miller:1960:IPF:321043.321046}. Here $u_i$ are continuous variables called \emph{node potentials} that indicate the visit order of node $i$ in the tour, while $p$ denotes the maximal number of nodes a vehicle can visit in a tour. This parameter can be used to specify the seat capability of vehicles. Generally $p$ won't exceed 10. Constraints (\ref{precedence}) ensure that a customer's origin precedes his destination. Constraints~(\ref{oneDriver-3}) ensure that each node is visited by exactly one vehciel. Constraints~(\ref{ODSameVehicle-3}) make sure that a customer's origin and destination are visited by the same vehicle. The intuitive meaning of constraints (\ref{bridge1-3}) and (\ref{bridge2-3}) is that if edge $(i, j)$ is selected in the solution then nodes $i, j$ must be served by the same vehicle. This set of constraints serve as a bridge between $\lambda$ variables (indicating whether an edge is selected) and $y$ variables (indicating whether a node is served by a particular vehicle). Here $M_1$ and $M_2$ are large numbers. 

\subsection{Coalition generation subproblem -- non-empty core}
%Recall that in Algorithm~\ref{algo: nucleolus} searching the most violated constraint in each iteration is a non-trivial task. 
Recall that when the RSP game has a non-empty core, the only coalitions that are non-redundant are the feasible coalitions (Proposition~\ref{prop: feasible}). Notice that $c(S), S\in \Ss$ is the minimum cost of a feasible route that covers the origin and destination of all the players in $S$, that is, those $s_i=1, i\in N$. This inspires us to formulate the subproblem explicitly. Let $G=(V, E)$ be the graph with vertex set $V=V_O \cup V_D \cup \{0\}$ and edge set $E=\{(i,j) \mid i,j \in V, i\neq j\}$. Here vertex $0$ is the ``dummy" depot such that any edge incident with it has a cost of $0$. $V_O(V_D)$ is the origin (destination) vertex set of players in $N$. Each player is associated with a profit (prize) equal to $y_i^*$. The constraint generation subproblem finds a feasible route in $G$ which maximizes the total prize minus cost, while conforming to the ridesharing problem constraints, such as capacity, precedence and maximum travel distance per passenger.

Denote the edge selection variable in the graph by $\lambda$. This problem is represented by
\begin{align}
	\label{obj}
	\text{($P_S^2$)} \quad \pi = \quad \text{max} \quad &\sum_{k\in N} y_{k}^*s_{k} - \sum_{i\in V}\sum_{j\in V} c_{ij}\lambda_{ij} + w^*  \\
	\label{capacity}
	\text{s.t.} \quad &\sum_{k\in N} s_k \le Q  \\
	\label{no-repeat}
	&\sum_{\{i \mid s_i^j=0\}} s_i + \sum_{\{i \mid s_i^j=1\}} (1 - s_i) \ge 1, \quad j \mid S_j \in \Omega \\
	%\label{profitable}
	%profitable constraints \\
	\label{startend}
	& \lambda_{0i} - \lambda_{(i+n)0} = 0, \quad i=1, 2, \ldots, n  \\
	\label{oneDriver-1}
	& \sum_{k\in N}\lambda_{0k} = 1 \\
	\label{oneDriver-2}
	& \sum_{k\in N}\lambda_{k+n,0} = 1 \\
	\label{balance-1-2}
	%& \sum_{i=0}^{k-1}\lambda_{ik} + \sum_{j=k+1}^{2n}\lambda_{kj} = 2s_k, \quad k\in N \\
	& \sum_{i=0}^{2n}\lambda_{ik} = s_k, \quad k\in N \\
	\label{balance-2}
	%& \sum_{i=0}^{k+n-1}\lambda_{i(k+n)} + \sum_{j=k+n+1}^{2n}\lambda_{(k+n)j} = 2s_k, \quad k\in N \\
	& \sum_{i=0}^{2n}\lambda_{ki} = s_k, \quad k\in N \\
	\label{balance-3}
	& \sum_{i=0}^{2n}\lambda_{i, k+n} = s_k, \quad k\in N \\
	\label{balance-4}
	& \sum_{i=0}^{2n}\lambda_{k+n, i} = s_k, \quad k\in N \\
	\label{SEC-2}
	& u_i - u_j + p\lambda_{ij} \le p-1, \quad 1\le i \neq j \le 2n  \\
	\label{precedence-2}
	& u_i < u_{i+n}, \quad i=1,2, \ldots, n \\
	& \lambda_{ij} \in \{0, 1\}, \quad i, j\in V, i\neq j \\
	& s_k\in \{0, 1\}, \quad k\in N  
\end{align}

Constraints~\eqref{no-repeat} put the restriction that a constraint that is generated before is not generated again. Constraints~\eqref{capacity} are the capacity constraints and \eqref{startend} forces a tour to start at origin and end at destination. Constraints~\eqref{balance-1-2},~\eqref{balance-2},~\eqref{balance-3} and~\eqref{balance-4} are the flow balancing constraints for each vertex. Constraints~\eqref{SEC-2} are the subtour elimination constraints and~\eqref{precedence-2} are the precedence constraints. 

%It is not hard to notice that this program is closely related to the RSP model as well as the explicit formulation $P_S^1$ in the previous subsection. Note that although related, $P_S^2$ is much easier to solve than $P_S^1$.
It is noteworthy that the propositions in Section~\ref{subsec: prop}  help reduce the complexity of Algorithm 1 significantly. This contribution is twofold. First, when the RSP game has a non-empty core, only feasible coalitions need to be considered in the subproblem ($P^2_S$). Second, although both prize-collecting RSP ($P^1_S$) and prize-collecting TSPPD ($P^2_S$) are NP-hard, prize-collecting RSP involves rider partitioning, therefore is much more difficult to solve than prize-collecting TSPPD, which concerns a single route with limited number of stops. 

\section{Experiments}
\label{sec: exp}

We have implemented the nucleolus algorithm (with both $P_S^1$ and $P_S^2$ as subproblem) in Java with CPLEX 12.6 and the Concert library. In this section, we first show the results of nucleolus algorithm with $P_S^2$ as subproblem. Because $P_S^2$ is much easier to solve than $P_S^1$, and the coalitions needed are much fewer in $P_S^2$ than in $P_S^1$, nucleolus algorithm with $P_S^2$ can not only calculate the nucleolus when the RSP has a non-empty core, but can also be used to find an approximate nucleolus when the corresponding RSP has an empty core. Next, we show a comparison between the nucleolus and the approximate nucleolus, obtained by using the nucleolus algorithm with $P_S^1$ and $P_S^2$ as the subproblem, respectively.

\subsection{Approximate nucleolus}
We report results for two instances of the 10-player problem, which is the largest problem we have solved. As will be shown later, the computational bottleneck is not at the nucleolus algorithm but at solving the corresponding ridesharing optimization problem (RSP). The data set\footnote{The data sets can be downloaded from \url{http://www.diku.dk/~sropke/}} we used in the experiments were selected from \cite{RopkeTSPPD}. The origins and destinations of customers were randomly generated in the square $[0, 1000]\times [0, 1000]$. The Euclidean distances were used. Table~\ref{table: prob10c} shows the geographical locations of the players. After solving the RSP MIP model, the optimal ridesharing plan is $\{1\}$, $\{3,5\}$, $\{4,6,7\}$, $\{8\}$, $\{2, 9\}$, $\{10\}$. These, along with other singleton coalitions are used to generate the initial constraints of the master LP problem. 

At stage $1$ three constraints are generated by solving the subproblem. They correspond to the coalitions of $\{3, 4, 6\}$, $\{4, 6\}$, and $\{2, 3, 4, 5, 9\}$. Active constraints corresponding to coalitions $\{8\}$, $\{10\}$ and $\{1\}$ are then identified at the end of stage $1$. At stage $2$, we notice that the solution to the master problem $(P_M^2)$ is unique, thus we find the approximate nucleolus. The approximate nucleolus of this game is listed in the last column of Table~\ref{table: prob10c}. %Note that after the pre-nucleolus is obtained, we notice that the lexicographically smallest excess vector is $0$. This means the core of this game is not empty and the pre-nucleolus is contained in the core. The pre-nucleolus of this game, in this case is also the nucleolus, is listed in the last column of Table~\ref{table: prob10c}. 

In total, $10+3+3=16$ out of $2^{10}-2 = 1022$ constraints are needed to compute the approximate nucleolus, which is only a very small fraction ($1.6\%$).   

\begin{table}[h]
\centering
\caption{Data and approximate nucleolus of prob10c}
\label{table: prob10c}
\begin{tabular}{
    S
    %@{(}
    >{{(}}
    S[table-format=3, table-space-text-pre=(]@{,}
    S[table-format=3, table-space-text-post=)]
    <{{)}} 
    %@{)\ }
    %@{\ \ (}
    >{{(}}
    S[table-format=3, table-space-text-pre=(]@{,}
    S[table-format=3, table-space-text-post=)] 
    %@{$)\ $}  
    <{{)}}
    S[table-format=3.1]
    %@{$)^t$}
}%{rlll}
\hline
\multicolumn{1}{c}{\begin{tabular}[c]{@{}l@{}}Customer\\ Number\end{tabular}} & \multicolumn{2}{c}{\begin{tabular}[c]{@{}l@{}}Pickup\\ Coordinates\end{tabular}} & \multicolumn{2}{c}{\begin{tabular}[c]{@{}l@{}}Drop-off\\ Coordinates\end{tabular}} & \multicolumn{1}{c}{\begin{tabular}[c]{@{}l@{}}Nucleolus\\ Cost\end{tabular}} \\ \hline
1                                                         & 387 & 137                                                   & 918 & 786                                                     & 346.2                                                    \\
2                                                         & 595 &     4                                                     & 852 & 236                                                     & 267.1                                                    \\
3                                                         & 514 & 483                                                   & 9 & 481                                                       & 627.5                                                    \\
4                                                         & 342 & 655                                                   & 609 & 55                                                      & 729.7                                                    \\
5                                                         & 715 & 887                                                   & 372 & 215                                                     & 434.3                                                    \\
6                                                         & 111  & 687                                                   & 777 & 91                                                      & 250.4                                                    \\
7                                                         & 692 & 933                                                   & 203 & 173                                                     & 97.5                                                     \\
8                                                         & 791  & 847                                                   & 488 & 312                                                     & 226.1                                                    \\
9                                                         & 702  & 762                                                   & 928 & 755                                                     & 520.8                                                    \\
10                                                       & 543  & 443                                                   & 90 & 700                                                      & 92.4                                                          \\ \hline
\end{tabular}
\end{table}

In our second experiment of prob10d (see Table~\ref{table: prob10d}), the optimal ridesharing configuration is $\{1\}$, $\{2, 3, 4, 6\}$, $\{7\}$, $\{5, 8\}$, $\{9\}$ and $\{10\}$. At stage 1, four constraints are generated after via solving the subproblem. They correspond to the coalitions of $\{2, 3, 4, 9\}$, $\{3, 6\}$, $\{2, 9\}$ and $\{2, 3, 4, 5, 6\}$. At stage $2$, the master LP problem is found to have a unique solution. This solution is thus the approximate nucleolus of this game. The approximate nucleolus is listed in the last column of Table~\ref{table: prob10d}. In total, $(10+2+4)/1022\approx 1.6\%$ constraints were needed to compute the approximate nucleolus. %This game again has a non-empty core, and the pre-nucleolus (in this case is also nucleolus) is listed in the last column of Table~\ref{table: prob10d}. In total, $(10+2+4)/1022\approx 1.6\%$ constraints were needed to compute the pre-nucleolus. 

\begin{table}[h]
\centering
\caption{Data and approximate nucleolus of prob10d}
\label{table: prob10d}
\begin{tabular}{
    S
    %@{(}
    >{{(}}
    S[table-format=3, table-space-text-pre=(]@{,}
    S[table-format=3, table-space-text-post=)]
    <{{)}} 
    %@{)\ }
    %@{\ \ (}
    >{{(}}
    S[table-format=3, table-space-text-pre=(]@{,}
    S[table-format=3, table-space-text-post=)] 
    %@{$)\ $}  
    <{{)}}
    S[table-format=3.1]
    %@{$)^t$}
}%{rlll}
\hline
\multicolumn{1}{c}{\begin{tabular}[c]{@{}l@{}}Customer\\ Number\end{tabular}} & \multicolumn{2}{c}{\begin{tabular}[c]{@{}l@{}}Pickup\\ Coordinates\end{tabular}} & \multicolumn{2}{c}{\begin{tabular}[c]{@{}l@{}}Drop-off\\ Coordinates\end{tabular}} & \multicolumn{1}{c}{\begin{tabular}[c]{@{}l@{}}Nucleolus\\ Cost\end{tabular}} \\ \hline
1                                                         & 60 & 742                                                    & 34 & 697                                                      & 52.0                                                     \\
2                                                         & 730 & 471                                                   & 390 & 845                                                     & 444.1                                                    \\
3                                                         & 964 & 151                                                   & 39 & 78                                                       & 808.7                                                    \\
4                                                         & 336 & 763                                                   & 11 & 332                                                      & 481.2                                                    \\
5                                                         & 330 & 593                                                   & 570 & 862                                                     & 326.9                                                    \\
6                                                         & 496 & 333                                                   & 88 & 346                                                      & 374.1                                                    \\
7                                                         & 168 & 403                                                   & 432 & 341                                                     & 271.2                                                    \\
8                                                         & 343 & 502                                                   & 525 & 846                                                     & 173.3                                                    \\
9                                                         & 600 & 534                                                   & 585 & 615                                                     & 83.0                                                     \\
10                                                        & 18 & 952                                                    & 494 & 605                                                     & 589.0                                                    \\ \hline
\end{tabular}
\end{table}

It is noteworthy to point out that the computational time for both instances is very small (less than 10s), indicating the bottleneck is the optimization solution method. %(recall this takes more than $100$ seconds).

\subsection{Nucleolus vs. approximate nucleolus}
In this subsection we conduct two experiments to compare the actual nucleolus and the approximate nucleolus. Finding the actual nucleolus by using Algorithm~\ref{algo: nucleolus} with $P_S^1$ as subproblem is a very time-consuming process. In our first example, problem8a, it takes $5$ hours to find the actual nucleolus. In our second example, problem8b, the time it takes to find the actual nucleolus goes up to $20$ hours.

The actual nucleolus cost and approximate nucleolus cost are summarized in Table~\ref{table: nuc-vs-appNuc-8a} and Table~\ref{table: nuc-vs-appNuc-8b}. As we can see, the solutions obtained by the approximate nucleolus algorithm are a close approximation for the actual nucleolus in both cases. 

It is of our interest to see the computational performance of Algorithm~\ref{algo: nucleolus} using $P_S^1$ and $P_S^2$. In problem8a, a total of $193$ constraints are generated by $P_S^1$, comparing to a total of $9$ constraints generated by $P_S^2$. In problem10d, a total of $184$ coalitions are generated by $P_S^1$, comparing to a total of only $9$ coalitions generated by $P_S^2$. Note that the total number of coalitions (not including the empty set and the universal set) is $2^8-2=254$ for problem8a and problem8b. Therefore, Algorithm~\ref{algo: nucleolus} using $P_S^1$ generated $193/254=76\%$ and $184/254=72\%$ of the total constraints to find the nucleolus in problem8a and problem8b, respectively. So Algorithm~\ref{algo: nucleolus} using $P_S^1$ is more like an enumeration procedure. Thus $P_S^2$ is much more computationally efficient than $P_S^1$, since the number of constraints generated by $P_S^1$ are significantly higher than that of $P_S^2$. This along with the fact that $P_S^1$ is much harder to solve than $P_S^2$ explains the significant time difference between Algorithm~\ref{algo: nucleolus} using $P_S^1$ and $P_S^2$. 

In Figure~\ref{solution-path-8a} and \ref{solution-path-8b} we measure the Euclidean distance between the incumbent nucleolus and the actual nucleolus ($\sqrt{\sum_{i\in N}(y_i^*-y_i)^2}$) as the algorithm iterates. These two figures show the solution path of the algorithm in both cases. As can be seen, in both cases, the algorithm found the nucleolus before it stopped. This happened before the 20th iteration in problem8a, and before $75$ constraints were generated in problem8b. It means the majority of the running time of this algorithm is consumed after the actual nucleolus is found. 
%\begin{table}[]
%\centering
%\caption{Nucleolus vs. approximate nucleolus}
%\label{table: nuc-vs-appNuc}
%\begin{tabular}{llll}
%\hline
%\begin{tabular}[c]{@{}l@{}}Customer \\ Number\end{tabular} & \begin{tabular}[c]{@{}l@{}}Nucleolus\\ Cost\end{tabular} & \begin{tabular}[c]{@{}l@{}}Approximate \\ Nucleolus Cost\end{tabular} & Difference \% \\ \hline
%1 & 408.1 & 408.1 & 0.0\% \\
%2 & 672.5 & 672.5 & 0.0\% \\
%3 & 287.6 & 287.6 & 0.0\% \\
%4 & 406.5 & 406.5 & 0.0\% \\
%5 & 888.2 & 888.2 & 0.0\% \\ \hline
%\end{tabular}
%\end{table}

\begin{table}[]
\centering
\caption{Nucleolus vs. approximate nucleolus -- problem8a}
\label{table: nuc-vs-appNuc-8a}
%\begin{tabular}{llll}
%\hline
%\begin{tabular}[c]{@{}l@{}}Customer \\ Number\end{tabular} & \begin{tabular}[c]{@{}l@{}}Nucleolus\\ Cost\end{tabular} & \begin{tabular}[c]{@{}l@{}}Approximate \\ Nucleolus Cost\end{tabular} & Difference \% \\ \hline
%1 & 442.1 & 449.9 & +1.8\% \\
%2 & 505.0 & 512.8 & +1.5\% \\
%3 & 639.2 & 636.6 & -0.4\% \\
%4 & 409.3 & 406.7 & -0.6\% \\
%5 & 527.6 & 525.0 & -0.5\% \\
%6 & 465.8 & 463.2 & -0.6\% \\
%7 & 228.4 & 225.8 & -1.1\% \\
%8 & 519.4 & 516.8 & -0.5\% \\ \hline
%\end{tabular}
\begin{tabular}{
	SSS
}
\hline
\multicolumn{1}{c}{\begin{tabular}[c]{@{}l@{}}Customer \\ Number\end{tabular}} & \multicolumn{1}{c}{\begin{tabular}[c]{@{}l@{}}Nucleolus\\ Cost\end{tabular}} & \multicolumn{1}{c}{\begin{tabular}[c]{@{}l@{}}Approximate \\ Nucleolus Cost\end{tabular}} \\ \hline
1 & 442.1 & 449.9 \\
2 & 505.0 & 512.8 \\
3 & 639.2 & 636.6 \\
4 & 409.3 & 406.7 \\
5 & 527.6 & 525.0 \\
6 & 465.8 & 463.2 \\
7 & 228.4 & 225.8 \\
8 & 519.4 & 516.8 \\ \hline
\end{tabular}
\end{table}

\begin{table}[]
\centering
\caption{Nucleolus vs. approximate nucleolus -- problem8b}
\label{table: nuc-vs-appNuc-8b}
%\begin{tabular}{llll}
%\hline
%\begin{tabular}[c]{@{}l@{}}Customer \\ Number\end{tabular} & \begin{tabular}[c]{@{}l@{}}Nucleolus\\ Cost\end{tabular} & \begin{tabular}[c]{@{}l@{}}Approximate \\ Nucleolus Cost\end{tabular} & Difference \% \\ \hline
%1 & 366.5 & 366.5 & 0.0\% \\
%2 & 507.1 & 457.7 & -9.7\% \\
%3 & 594.1 & 617.1 & +3.9\% \\
%4 & 387.6 & 456.7 & +17.8\% \\
%5 & 1012.5 & 1035.6 & +2.3\% \\
%6 & 258.1 & 235.1 & -8.9\% \\
%7 & 545.3 & 571.8 & +4.9\% \\
%8 & 235.9 & 166.8 & -29.3\% \\ \hline
%\end{tabular}
\begin{tabular}{
	S	
	S
	S
}
\hline
\multicolumn{1}{c}{\begin{tabular}[c]{@{}l@{}}Customer \\ Number\end{tabular}} & \multicolumn{1}{c}{\begin{tabular}[c]{@{}l@{}}Nucleolus\\ Cost\end{tabular}} & \multicolumn{1}{c}{\begin{tabular}[c]{@{}l@{}}Approximate \\ Nucleolus Cost\end{tabular}} \\ \hline
1 & 366.5 & 366.5 \\
2 & 507.1 & 457.7 \\
3 & 594.1 & 617.1 \\
4 & 387.6 & 456.7 \\
5 & 1012.5 & 1035.6 \\
6 & 258.1 & 235.1 \\
7 & 545.3 & 571.8 \\
8 & 235.9 & 166.8 \\ \hline
\end{tabular}
\end{table}

\begin{figure}[ht]
\centering
\includegraphics[width=4in]{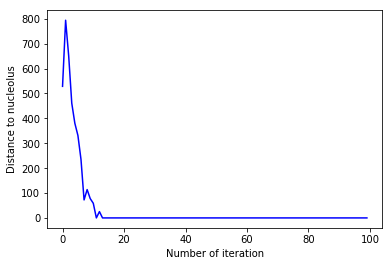}
\caption{Solution path -- problem8a}
\label{solution-path-8a}
\end{figure}
%\end{frame}

%\begin{frame}
%\frametitle{Nucleolus vs. approximate nucleolus}
\begin{figure}[ht]
\centering
\includegraphics[width=4in]{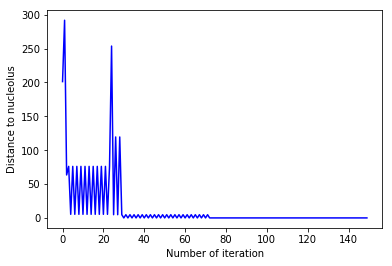}
\caption{Solution path -- problem8b}
\label{solution-path-8b}
\end{figure}
%\end{frame}

%\begin{table}[]
%\centering
%\caption{Nucleolus vs. approximate nucleolus -- problem10d}
%\label{table: nuc-vs-appNuc-10}
%\begin{tabular}{llll}
%\hline
%\begin{tabular}[c]{@{}l@{}}Customer \\ Number\end{tabular} & \begin{tabular}[c]{@{}l@{}}Nucleolus\\ Cost\end{tabular} & \begin{tabular}[c]{@{}l@{}}Approximate \\ Nucleolus Cost\end{tabular} & Difference \% \\ \hline
%1 & 52.0 & 52.0 & 0.0\% \\
%2 & 444.1 & 444.1 & 0.0\% \\
%3 & 718.8 & 808.7 & +12.5\% \\
%4 & 537.1 & 481.2 & -10.4\% \\
%5 & 326.9 & 326.9 & 0.0\% \\
%6 & 408.2 & 374.1 & -8.4\% \\
%7 & 271.2 & 271.2 & 0.0\% \\
%8 & 173.3 & 173.3 & 0.0\% \\ 
%9 & 83.0 & 83.0 & 0.0\% \\
%10 & 589.0 & 589.0 & 0.0\% \\ \hline
%\end{tabular}
%\end{table}

%\section{Literature Review}
%\cite{Engevall:2004:HVG:989115.989152} modeled a cost-allocation problem that arises in a logistics application as a heterogeneous vehicle-routing game (VRG).

%\section{Stable Issues in Ridesharing}

%\section{Fair cost allocation}

\section{Conclusions}
\label{sec: conclusions}
We studied an important problem faced by ridesharing service provider: how to allocate cost among ridesharing participants to ensure sustainability and fairness. This fair cost allocation problem was modeled as a cooperative game. A special property of the cooperative ridesharing game is that its characteristic function values are calculated by solving an optimization problem. To better understand this game, we further studied the characteristic function and proved it to be monotone, subadditive, but non-convex (meaning the core can be empty). The most fair allocation plan is identified by the nucleolus of the RSP game. We then proposed an iterative constraint-generation algorithm (Algorithm~\ref{algo: nucleolus}) for calculating it in two situations -- the game has an empty core and the game has a non-empty core. In both cases the algorithm utilizes an explicitly formulated MIP as the subproblem to generate constraints. When the game has an empty core, the algorithm uses $P_S^1$ as the subproblem and becomes an enumeration procedure to find the nucleolus of the game. When the game has a non-empty core, this algorithm uses $P_S^2$ as the subproblem which utilizes the special properties of the RSP game such that the characteristic function values are computed only when they are needed. Therefore the number of subproblems (an NP-hard optimization problem) that need to be solved is significantly reduced. Experiments showed that by adopting this algorithm with $P_S^2$ only a small fraction ($1.6\%$) of the coalition constraints were needed to find the nucleolus. It is also found in the experiments where the emptiness of the RSP game is unclear, the algorithm with $P_S^2$ can be used to find an approximate nucleolus that is close to the actual one. This indicates that our proposed algorithm is promising in finding nucleolus of dynamic, large-scale RSP game and that since a cooperative game theory modeling approach does not necessarily differentiate drivers and riders explicitly, our model, the mathematical programs and the algorithm proposed in this paper also have very promising application in an autonomous vehicle ridesharing systems. We can see a few interesting and promising future research directions related to this study. First, efficient heuristics for the subproblems can be developed. Second, with such efficient algorithms in hand, we can further investigate the interaction between the vehicle capacity and the cost allocation, which we believe will provide insight on the intrinsic nature of ridesharing game. Besides, we started our study with the motivation of designing mathematical models and algorithms for the most general case of ridesharing scenario with the least amount of assumptions on the coalition etc. However in the real world, the existence of certain situations can significantly simplify the calculation. For instance, in the case some of the ridesharing participants have formed a coalition on their own, we can exploit these structural properties to simplify the coalition generation scheme and expedite the nucleolus calculation.

\section*{ACKNOWLEDGEMENTS}
The authors would gratefully acknowledge the kind support from the Dissertation Fellowship of Texas A\&M University. We are grateful to anonymous reviewers whose valuable suggestions have led to a considerable improvement in the organization and presentation of this manuscript. 
%The authors are very grateful to the two referees for their insightful comments and suggestions. 
%The authors would gratefully acknowledge the kind support from the National Center for Freight and Infrastructure Research and Education (CFIRE) at the University of Wisconsin-Madison and the Southwest Region University Transportation Center (SWUTC). 

\bibliographystyle{plainnat}
\bibliography{bibReference}

%\lipsum[1]

\section*{Appendix}
We show the RSP game has the following properties. From here on we denote by $C(\cdot)$ the mathematical program that defines the characteristic function value of $\cdot$, i.e. $c(\cdot)$.

\begin{proposition}[Monotonicity]
	The characteristic function of the RSP game is \emph{monotone}, that is, $c(S)\le c(T), S\subset T\subset N$.
\end{proposition}

\begin{proof}
	Proof by contradiction. Suppose there exists $S\subset T\subset N$ and $c(S)>c(T)$. Let $X=\{x_{r} \mid r\in R\}$ be an optimal solution to $C(T)$. Let $R_T=\{r_i \mid x_{r_i}=1\}$. 
    %Without loss of generality, we assume $S=\{1,\ldots,i\}$ and $T=\{1, \ldots, i, i+1, \ldots, j\}$. Let $R_T = R_{T_1} \cup R_{T_2}$ where $R_{T_1}$ is the set of singleton routes, that is a route that only covers one player, and $R_{T_2}$ is the set of non-singleton routes. 
    
	For $r\in R$, we construct a feasible solution to $C(S)$ in the following manner. Let 
	% \begin{itemize}
	% 	\item $x_r = 0, \forall r\in R_{T_1}$, 
	% 	\item $x_r = 1, \forall r\in R_{T_2}$,
	% 	\item $X' = \{x_r \mid r\in R_\}$
	% \end{itemize}
	\begin{equation*}
	x_r = \begin{dcases}
		1, \quad \text{if $\exists i\in S$ such that $a_{ir}=1$}, \\
		0, \quad \text{otherwise}. %\quad \text{if $\forall i\in S$ such that $a_{ir}=0$},
	      \end{dcases}
	\end{equation*}	

	Let $X'$ be the solution constructed in the above way. Denote by $R_S$ the set of selected routes. Intuitively, we keep those routes in $R_T$ that covers at least one player in $S$ and discard those don't. 

	It is known that $X$ must satisfy 
	\begin{align*}
		&\sum_{r\in R} a_{ir}x_{r}= 1,  \quad i\in T \\
		&\sum_{r\in R} a_{ir}x_{r}= 0,  \quad i\in N-T
	\end{align*}
	Because $S\subset T$, then $X'$ must satisfy
	\begin{align*}
		&\sum_{r\in R} a_{ir}x_{r}= 1,  \quad i\in S \\
		&\sum_{r\in R} a_{ir}x_{r}= 0,  \quad i\in T-S \\
		&\sum_{r\in R} a_{ir}x_{r}= 0,  \quad i\in N-T
	\end{align*}
	This is equivalent to
	%Further notice that $R_T\subset R$, then $X$ must satisfy 
	\begin{equation*}
		\sum_{r\in R} a_{ir}x_{r}= s_i,  \quad i\in N
	\end{equation*}
	So $X'$ is a feasible solution to $C(S)$. In addition, since the cost matrix $\{c_{ij}\}$ is positive, the route cost is also positive. Therefore the cost of $X'$ is less than $c(T)$, which is less than $c(S)$. Note that $c(S)$ is the optimal cost, so this is a contradiction.  \qed
\end{proof}

\begin{proposition}[Subadditivity]
	The characteristic cost function of RSP game is \emph{subadditive}, i.e., $c(S)+c(T)\ge c(S\cup T), S,T\subset N, S\cap T = \emptyset$.
\end{proposition}

\begin{proof}
	Let $R_S$, $R_T$ be the optimal solution to $C(S)$, $C(T)$, respectively. Because $S, T\subset N$ and $S\cap T = \emptyset$, $R_{S,T} = R_S\cup R_T$ must cover all the players in $S\cup T$, i.e. $R_{S, T}$ is a feasible solution to $C(S\cup T)$. Since this solution has an objective value equal to $c(S)+c(T)$, we have $c(S)+c(T)\ge c(S\cup T)$. \qed
\end{proof}

\end{spacing}

\end{document}